\documentclass{article}
\usepackage{ijcai15}

\usepackage{times}
\usepackage[boxed]{algorithm}
\usepackage[noend]{algorithmic}
\usepackage{paralist}
\usepackage[centertags,fleqn]{amsmath}
\usepackage{amssymb,amsthm,xspace}
\usepackage{booktabs}
\usepackage{nicefrac}
\usepackage{tabularx}
\usepackage{mathtools}
\usepackage{paralist}
\usepackage{balance}

\usepackage[textwidth=15mm,textsize=footnotesize]{todonotes}

\newtheorem{theorem}{Theorem}

\newtheorem{corollary}{Corollary}
\newtheorem{example}{Example}

\renewcommand{\algorithmicrequire}{\textbf{Input:}}
\renewcommand{\algorithmicensure}{\textbf{Output:}}
\algsetup{linenodelimiter=\,}
\algsetup{linenosize=\tiny}
\algsetup{indent=2em}

	\usepackage{subfig}
	\usepackage{colortbl,color,array,hhline}
	\usepackage{tikz}
	\usetikzlibrary{arrows}
	\usetikzlibrary{decorations.pathreplacing}
	\usepackage{lscape}

	\DeclareSymbolFont{largesymbolsA}{U}{txexa}{m}{n}

	\usepackage{amsmath,amssymb,amsfonts}

	\usetikzlibrary{arrows}
	\usetikzlibrary{decorations.pathreplacing}
	\usetikzlibrary{patterns}

	\newcommand{\pref}{\succ \xspace}

	\newcommand{\Indiff}[1][]{
		\ifthenelse{\equal{#1}{}}{\mathrel I}{\mathop{I_{#1}}}
	}
	\newcommand{\prefset}[1][]{\ifthenelse{\equal{#1}{}}{\mathcal{R}}{\mathcal{R}_{#1}}}

	\usepackage{enumitem}
	\setenumerate[1]{label=\rm(\it{\roman{*}}\rm),ref=({\it\roman{*}}),leftmargin=*}

	\usepackage[mathcal]{euscript}

			\newcommand{\sPref}[1][]{                  
				\ifthenelse{\equal{#1}{}}{\mathrel P}{\mathop{P_{#1}}}
			}

			\tikzset{
			  treenode/.style = {align=center, inner sep=0pt, text centered,
			    font=\sffamily},
			  arn_n/.style = {treenode, circle, white, font=\sffamily\bfseries, draw=black,
			    fill=black, text width=1.5em},
			  arn_r/.style = {treenode, circle, red, draw=red,
			    text width=1.5em, very thick},
			  arn_x/.style = {treenode, rectangle, draw=black,
			    minimum width=0.5em, minimum height=0.5em}
			}

			\newcommand{\midd}{\mathbin{:}}
			\newcommand{\reals}{\ensuremath{\mathbb{R}}}
			
			\newcommand{\citet}[1]{\citeauthor{#1}~\shortcite{#1}}
			\newcommand{\citep}{\cite}

\definecolor{kentuckyblue}{RGB}{0, 93, 170}			
\definecolor{green}{RGB}{0, 102, 0}					
\definecolor{frenchred}{RGB}{250,60,50}				
\definecolor{unswyellow}{RGB}{255,155,0}			
\definecolor{bavblue}{RGB}{1,153,213}				

\begin{document}

\title{Equilibria Under the Probabilistic Serial Rule}
%

\author{Haris Aziz \and  Serge Gaspers  \and Simon Mackenzie \and Nicholas Mattei\\
NICTA and UNSW,
Sydney, Australia \\
\{haris.aziz, serge.gaspers, simon.mackenzie, nicholas.mattei\}@nicta.com.au\\
\AND Nina Narodytska\\ Carnegie Mellon University\\
ninan@gmail.com
\And Toby Walsh\\
NICTA and UNSW,
Sydney, Australia\\
toby.walsh@nicta.com.au
}

\maketitle
\begin{abstract}
The probabilistic serial (PS) rule is a prominent randomized rule for assigning
indivisible goods to agents. Although it is well known for its good fairness and welfare properties, it is not strategyproof. In view of this, we address several fundamental questions regarding equilibria under PS. Firstly, we show that Nash deviations under the PS rule can cycle. Despite the possibilities of cycles, we prove that a pure Nash equilibrium is guaranteed to exist under the PS rule. We then show that verifying whether a given profile is a pure Nash equilibrium is coNP-complete, and computing a pure Nash equilibrium is NP-hard. For two agents, we present a linear-time algorithm to compute a pure Nash equilibrium which yields the same assignment as the truthful profile. Finally, we conduct experiments to evaluate the quality of the equilibria that exist under the PS rule, finding that the vast majority of pure Nash equilibria yield social welfare that is at least that of the truthful profile. 
\end{abstract}

		\section{Introduction}
		



Resource allocation is a fundamental and widely applicable area within AI and computer
science. When resource allocation rules are not strategyproof and agents do not have incentive to report their preferences truthfully, it is important to understand the possible manipulations; Nash dynamics; and the existence and computation of equilibria.

In this paper we consider the  \emph{probabilistic serial (PS)} rule for the \emph{assignment problem}. 
In the \emph{assignment problem} we have a possibly unequal number of agents and objects where
the agents express preferences over 
objects and, based on these preferences, the objects are allocated to the agents~\citep{AGMW14a,BoMo01a,Gard73b,HyZe79a}. The model is applicable to many 
resource allocation and fair division settings where the objects may be public houses, school seats, course enrollments, kidneys for transplant, car park spaces, chores, joint assets, or time slots in schedules. 
The \emph{probabilistic serial (PS)} rule is a randomized (or fractional) assignment rule.
A randomized or fractional assignment rule takes the preferences of the agents into account in order to allocate each agent a fraction of the object. If the objects are indivisible but allocated in a randomized way, the fraction can also be interpreted as the probability of receiving the object. 
Randomization is widespread in resource 
allocation as it is a natural way to ensure procedural fairness~\citep{BCKM12a}.

A prominent randomized assignment rule is the PS rule~\citep{BoHe12a,BoMo01a,BCKM12a,KaSe06a,Koji09a,Yilm10a,SaSe13b}.  
PS works as follows: each agent expresses a linear order over the set of houses.\footnote{We use 
the term house throughout the paper though we stress any object could be allocated with these mechanisms.}
Each house is considered to have a divisible probability weight of one. Agents simultaneously and at 
the same speed eat the probability weight of their most preferred house that has not yet been completely
eaten. Once a house has been 
completely eaten by a subset of the agents, each of these agents starts eating his next most preferred 
house that has not been completely eaten (i.e., they may ``join'' other agents already eating a different
house or begin eating new houses). The procedure terminates after all the houses have been completely eaten. 
The random allocation of an agent by PS is the amount of each object he has eaten. 
Although PS was originally defined for the setting where the number of houses is equal to the 
number of agents, it can be used without any modification for any number of  
houses relative to the number agents~\citep{BoMo01a,Koji09a}.

In order to compare random allocations, an agent needs to consider relations between them. 
We consider two well-known relations between random allocation~\cite{ScVa12a,SaSe13b,Cho12a}:
$(i)$ \textit{expected utility (EU)}, 
and $(ii)$ \textit{downward lexicographic (DL)}. 
For EU, an agent prefers an allocation that yields more expected utility. 
For DL, an agent prefers an allocation that gives a higher probability to the most preferred 
alternative that has different probabilities in the two allocations. 
Throughout the paper, we assume that agents express \emph{strict} preferences over houses, 
i.e., they are not indifferent between any two houses. 

The PS rule fares well in terms of 
fairness and welfare~\citep{BoHe12a,BoMo01a,BCKM12a,Koji09a,Yilm10a}.
It satisfies strong envy-freeness and efficiency with respect to the DL 
relation~\citep{BoMo01a,ScVa12a,Koji09a}. Generalizations of the PS rule have been 
recommended and applied in many settings~\citep{AzSt14a,BCKM12a}.  The PS rule also satisfies some 
desirable incentive properties: if the number of objects is at most the number of agents, then PS is  
DL-strategyproof~\citep{BoMo01a,ScVa12a}. 
Another well-established rule, \emph{random serial dictator (RSD)}, is not envy-free, not as efficient 
as PS~\citep{BoMo01a}, and the fractional allocations under RSD are \#P-complete 
to compute~\citep{ABB13b}. 

Although PS performs well in terms of fairness and welfare, unlike RSD, it is not strategyproof. \citet{AGM+15c} showed
that, in the scenario where one agent is strategic, computing his best
response (manipulation) under complete information of the other agents' strategies is NP-hard for 
the EU relation, but polynomial-time computable for the DL relation.
In this paper, we consider the situation where \emph{all} agents are strategic. We especially focus on pure Nash equilibria (PNE) --- reported preferences profiles for which no agent has an incentive to report a different preference.
We examine the following natural questions for the first time: \emph{(i) What is the nature of best response dynamics under the PS rule? (ii) Is a (pure) Nash equilibrium always guaranteed to exist? (iii) How efficiently can a (pure) Nash equilibrium be computed?  (iv) What is the difference in quality of the various equlibria that are possible under the PS rule?} In related work,
\citet{EkKe12a} showed that when agents are not truthful, the outcome of 
PS may not satisfy desirable properties related to efficiency and envy-freeness. 
\citet{HeMa12a} provided a necessary and sufficient condition for implementability of Nash equilibrium  for the random assignment problem. 

\bigskip
\noindent
\textbf{Contributions.\;} 
For the PS rule we show that expected utility best responses can cycle for \emph{any} cardinal utilities consistent with the ordinal preferences. This is significant as Nash dynamics in matching theory 
has been an active area of research, especially for the stable matching problem~\cite{AGM+11a}, and the 
presence of a cycle means that following a sequence of best responses is not guaranteed to result in an 
equilibrium profile.
We then prove that a pure Nash equilibrium (PNE) is guaranteed to exist for any number of agents and houses and any utilities.
To the best of our knowledge, this is the first proof of the existence of a Nash equilibrium for the PS rule. 
For the case of two agents we present a linear-time algorithm to compute a preference profile that is in PNE with respect to the original preferences.
We show that the general problem for computing a PNE is NP-hard. Finally, we run a set of 
experiments on real and synthetic preference data to evaluate the welfare achieved by PNE profiles compared to the welfare achieved under the truthful profile. 

		\section{Preliminaries}


		An \emph{assignment problem} $(N, H, \pref)$ consists  of a set of agents $N=\{1,\ldots, n\}$, a set of houses $H=\{h_1, \ldots, h_m\}$ and a preference profile $\pref=(\pref_1,\ldots, \pref_n)$ in which $\pref_i$ denotes a complete, transitive and strict ordering on $H$ representing the preferences of agent $i$ over the houses in  $H$.
			A \emph{fractional assignment} is an $(n\times m)$ matrix $[p(i)(h_j)]_{\substack{1\leq i\leq n, 1\leq j\leq m}}$ such that for all $i\in N$, and $h_j\in H$, $0\leq p(i)(h_j)\leq 1$;  and for all $j\in \{1,\ldots, m\}$, $\sum_{i\in N}p(i)(h_j)= 1$. 
		The value $p(i)(h_j)$ is the fraction of house $h_j$ that agent $i$ gets. Each row $p(i)=(p(i)(h_1),\ldots, p(i)(h_m))$ represents the \emph{allocation} of agent $i$.
	A fractional assignment can also be interpreted as a random assignment where $p(i)(h_j)$ is the probability of agent $i$ getting house $h_j$.


			Given two random assignments $p$ and $q$, $p(i) \pref_i^{DL} q(i)$ i.e.,  a player $i$ \emph{DL~(downward lexicographic) prefers}  allocation $p(i)$ to $q(i)$ if $p(i)\neq q(i)$ and for the most preferred house $h$ such that $p(i)(h)\neq q(i)(h)$, we have that $p(i)(h)>q(i)(h)$.
When agents are considered to have cardinal utilities for the objects, we denote by $u_i(h)$ the utility that agent $i$ gets from house $h$. We will assume that the total utility of an agent equals the sum of the utilities that he gets from each of the houses. Given two random assignments $p$ and $q$, $p(i) \pref_i^{EU} q(i)$, i.e.,  a player $i$ \emph{EU (expected utility)~prefers} allocation $p(i)$ to $q(i)$ if
			$\sum_{h\in H}u_i(h) \cdot p(i)(h)> \sum_{h\in H}u_i(h) \cdot q(i)(h).$
		Since for all $i\in N$, agent $i$ compares assignment $p$ with assignment $q$ only with respect to his allocations $p(i)$ and $q(i)$, we will sometimes abuse the notation and use $p\pref_i^{EU} q$ for $p(i)\pref_i^{EU} q(i)$.

		A \emph{random assignment rule} takes as input an assignment problem $(N,H,\pref)$ and returns a random assignment which specifies what fraction or probability of each house is allocated to each agent.
We will primarily focus on the expected utility setting but will comment 
on and use DL wherever needed.

\paragraph{The Probabilistic Serial Rule and Equilibria.}
The \emph{Probabilistic Serial (PS) rule} is a random assignment algorithm in which we consider each house as infinitely divisible.
At each point in time, each agent is eating (consuming the probability mass of) his most preferred house that has not been completely eaten. Each agent eats at the same unit speed. Hence all the houses are eaten at time $m/n$ and each agent receives a total of $m/n$ units of houses.
The probability of house $h_j$ being allocated to $i$ is the fraction of house $h_j$ that $i$ has eaten. The PS fractional assignment can be computed in time $O(mn)$.
We refer the reader to \citet{BoMo01a} or \citet{Koji09a} for alternative definitions of PS. 
 The following example from \citet{BoMo01a,AGM+15c} shows how PS works.

\begin{example}[PS rule]\label{example:PS}
	Consider an assignment problem with the following preference profile.
\begin{align*}
	\centering
	\succ_1:&\quad h_1,h_2,h_3 \\ \succ_2:&\quad h_2,h_1,h_3 \\ \succ_3:&\quad  h_2,h_3,h_1
	\end{align*}
	Agents $2$ and $3$ start eating $h_2$ simultaneously whereas agent $1$ eats $h_1$. When $2$ and $3$ finish $h_2$, agent $3$ has only eaten half of $h_1$.  The timing of the eating can be seen below.
\begin{center}
             \begin{tikzpicture}[scale=0.2]
                 \centering
                 \draw[-] (0,0) -- (0,6);
                 \draw[-] (0,0) -- (20,0);

                 \draw[-] (20,6) -- (20,0);

\draw[-] (0,2) -- (20,2);
\draw[-] (0,4) -- (20,4);
\draw[-] (20,0) -- (20,6);

\draw[-] (10,0) -- (10,6);

\draw[-] (0,6) -- (20,6);

\draw[-] (15,0) -- (15,6);

                                        \draw (0,-.8) node(c){\small $0$};
                             \draw (20/2,-1.2) node(c){\small $\frac{1}{2}$};

 \draw (20/2,-2.5) node(c){\small Time};

                             \draw (20,-1) node(c){\small$1$};

\draw (15,-1.2) node(c){\small$\frac{3}{4}$};

    \draw(-3,6) node(z){\small Agent $1$};
                 \draw(-3,4) node(z){\small Agent $2$};
                 \draw(-3,2) node(z){\small Agent $3$};

\draw(5,6.8) node(z){\small $h_1$};

\draw(5,4.8) node(z){\small $h_2$};

\draw(5,2.8) node(z){\small $h_2$};

\draw(12.5,6.8) node(z){\small $h_1$};

\draw(12.5,4.8) node(z){\small $h_1$};

\draw(12.5,2.8) node(z){\small $h_3$};

\draw(17.5,6.8) node(z){\small $h_3$};

\draw(17.5,4.8) node(z){\small $h_3$};

\draw(17.5,2.8) node(z){\small $h_3$};
  \end{tikzpicture}
\end{center}

\noindent
	The final allocation computed by PS is
\[
PS(\succ_1,\succ_2,\succ_3)=\begin{pmatrix}
	3/4&0&1/4\\
  1/4&1/2& 1/4 \\
  0&1/2 &  1/2

 \end{pmatrix}.
\]

\end{example}

Consider the assignment problem in Example~\ref{example:PS}. If agent $1$ misreports his preferences as follows: $\succ_1':\quad h_2,h_1,h_3,$ then \[
		PS(\succ_1',\succ_2,\succ_3)=\begin{pmatrix}
			1/2&1/3&1/6\\
		  1/2 & 1/3 & 1/6 \\
		  0 & 1/3 & 2/3
			\end{pmatrix}.
		\]
		\noindent
		If we suppose that $u_1(h_1)=7$, $u_1(h_2)=6$, and $u_1(h_3)=0$, then agent $1$ gets more expected utility when he reports $\succ_1'$. In the example,
the truthful profile is in PNE with respect to DL preferences but not expected utility.

 We study the existence and computation of Nash equilibria.
For a preference profile $\pref$, we denote by $(\pref_{-i},\pref_i')$ the preference profile obtained from $\pref$ by replacing agent $i$'s preference by $\pref_i'$.


\section{Nash Dynamics}




When considering Nash equilibria of any setting, one of the most natural ways of proving that a PNE always exists is to show that better or best responses do not cycle which implies that eventually, Nash dynamics terminate at a Nash equilibrium profile. Our first result is that DL and EU best responses can cycle. For EU best responses, this is even the case when agents have Borda utilities.  

\begin{theorem}\label{th:cycle}
With 2 agents and 5 houses where agents have Borda utilities, 
EU best responses can lead to a cycle in the profile.
\end{theorem}

\begin{proof}
The following 5 step sequence of best responses leads
to a cycle.  We use $U$ to denote
the matrix of utilities of the agents over the houses such that
$U[1,1]$ is the utility of agent $1$ for house $h_1$.
Note that $P$ starts as the truthful reporting in our example.
The initial preferences and utilities of the agents are:

\noindent
\begin{minipage}{.48\columnwidth}
\begin{align*}
	\succ_1:\quad & h_2,h_3,h_5,h_4,h_1  \\
	\succ_2:\quad & h_5,h_3,h_4,h_1,h_2  
 \end{align*}
\end{minipage}
\begin{minipage}{.48\columnwidth}
\[
 U_0 = \begin{pmatrix}
	0 & 4 & 3 & 1 & 2 \\
	1 & 0 & 3 & 2 & 4  
 \end{pmatrix}.
\]
\end{minipage}

\smallskip
\noindent
This yields the following allocation and utilities at the start:

\noindent
\[
 PS(\succ_1,\succ_2) = \begin{pmatrix}
	1/2 	& 1	 & 1/2	 & 1/2	& 0	 \\
	1/2 	& 0	 & 1/2	 & 1/2	& 1	 \\
 \end{pmatrix},  EU_0 = \begin{pmatrix}
	6  \\
 	7 
 \end{pmatrix}.
\]

\noindent
In Step 1, agent 1 deviates to increase his utility.
He reports the preference $\succ_1': h_3,h_4,h_2,h_1,h_5$;
which results in
\[
 PS(\succ_1',\succ_2) = \begin{pmatrix}
	0 	& 1	 & 1	 	& 1/2		& 0	 \\
	1 	& 0	 & 0		& 1/2		& 1	 \\
 \end{pmatrix}, EU_1 = \begin{pmatrix}
	7.5  \\
 	6 
 \end{pmatrix}.
\]

\noindent
In Step 2, agent 2 changes his report to
$\succ_2': h_3,h_4,h_5,h_1,h_2.$
This increases his utility to 7 and decreases
the utility of agent 1 to 6.

In Step 3, Agent 1 changes his report
to $\succ_1'': h_3,h_5,h_2,h_1,h_4.$ This 
increases the utility of agent 1 to 7.5 and
decreases the utility of agent 2 to 4.5.

In Step 4, Agent 2 changes his report to
$\succ_2'': h_5,h_3,h_4,h_1,h_2.$ which 
increases his expected utility to 6.5 while 
decreasing the expected utility of agent 1
to 7.

In Step 5, Agent 1 changes his report
to $\succ_1''': h_3, h_4, h_2, h_1, h_5.$
Notice that $\succ_1''' = \succ_1'$ and 
$\succ_2'' = \succ_2$.
%
%
This is the same profile as the one of Step 1, so we have 
cycled. 
\end{proof}
It can be verified that every response in this example is both an EU best response (with respect to any cardinal utilities consistent with the ordinal preferences) and also 
a DL best response.  Hence,  DL best responses and EU best responses (with respect to any cardinal utilities consistent with the ordinal preferences) can cycle.

The fact that best responses can cycle means that simply following best responses need not result in a PNE. Hence the normal form game induced by the PS rule is not a potential game~\citep{MoSh96a}.
%
Checking whether an instance has a Nash equilibrium appears to be a challenging problem. The naive method requires going through $O({m!}^n)$ profiles, which is super-polynomial even when $n=O(1)$ or $m=O(1)$.

\section{Existence of Pure Nash Equilibria}
%
Although it seems that computing a Nash equilibrium is a challenging problem (we give hardness results in the next section),
we show that at least one (pure) Nash equilibrium is guaranteed to exist for any number of houses, any number of agents, and any preference relation over fractional allocations.\footnote{We already know from Nash's original result that a \emph{mixed} Nash equilibrium exists for any game.} The proof relies on showing that the PS rule can be modelled as a perfect information extensive form game. 	

\begin{theorem}
A PNE is guaranteed to exist under the PS rule for any number of agents and houses, and for any relation between allocations.
\end{theorem}
		
		\begin{proof}[Proof sketch]

	Consider running PS on all possible ${m!}^n$ preference profiles for $n$ agents and $m$ objects. In each profile $i$, let $t_i^1,\dots, t_i^{k_i}$ be the $k_i$ different time points in the PS algorithm run for the $i$-th profile when at least one house is finished. Let $g=\text{GCD}(\{t_i^{j+1}-t_i^j\midd j\in \{1,\ldots, k_i-1\}, i\in \{1,\dots,m!^n\})$ where GCD denotes the greatest common divisor. Since in each profile $i$, $t_i^{j+1}-t_i^j>0$ for all $j\in \{0,\ldots, k_i-1\}$, we have that $g$ is finite and greater than zero.

	The time interval length $g$ is small enough such that each run of the PS rule can be considered to have $m/g$ stages of duration $g$. Each stage can be viewed as having $n$ sub-stages so that in each stage, agent $i$ eats $g/n$ units of a house in sub-stage $i$ of a stage.
		In each sub-stage only one agent eats $g/n$ units of the most favoured house that is available. Hence we now view PS as consisting of a total of $mn/g$ sub-stages and the agents keep coming in order $1,2,\ldots, n$ to eat $g$ units of the most preferred house that is still available. If an agent eats $g$ units of a house in a stage then it will eat $g$ units of the same house in his sub-stage of the next stage as long as the house has not been fully eaten.
		Consider a perfect information extensive form game tree.
		For a fixed reported preference profile, the PS rule unravels accordingly along a path starting at the root and ending at a leaf. Each level of the tree represents a sub-stage in which a certain agent has his turn to eat $g$ units of his most preferred available house. Note that there is a one-to-one correspondence between the paths in the tree and the ways the PS algorithm can be implemented, depending on the reported preference. 

		A subgame perfect Nash equilibrium (SPNE) is guaranteed to exist for such a game via backward induction:
		starting from the leaves and moving towards the root of the tree, the agent at the specific node chooses an action that maximizes his utility given the actions determined for the children of the node. The SNPE identifies at least one such path from a leaf to the root of the game. The path can be used to read out the most preferred house of each agent at each point. The information provided is sufficient to construct a preference profile that is in Nash equilibrium. Those houses that an agent did not eat at all can conveniently be placed at the end of the preference list. Such a preference profile is in Nash equilibrium. 
		\end{proof}

\section{Complexity of Pure Nash Equilibrium}


Our argument for the existence of a Nash equilibrium is constructive. However, naively constructing the extensive form game and then computing a subgame perfect Nash equilibrium requires exponential space and time. It is unclear whether a sub-game perfect Nash equilibrium or any Nash equilibrium preference profile can be computed in polynomial time.

\subsection{General Complexity Results}

In this section, we show that computing a PNE is NP-hard and verifying whether a profile is a PNE is coNP-complete. Recently it was shown that computing an expected utility best response is NP-hard~\citep{AGM+15b,AGM+15c}. Since equilibria and best responses are somewhat similar, one would expect that problems related to equilibria under PS are also hard. However, there is no general reduction from best response to equilibria computation or verification. In view of this, we prove results regarding PNE by closely analyzing the reduction in \citep{AGM+15b}.
First, we show that checking whether a given preference profile is in PNE under the PS rule is coNP-complete.

\begin{theorem}
	Given agents' utilities, checking whether a given preference profile is in PNE under the PS rule is coNP-complete.
	
	\end{theorem}
	
												\begin{proof}[Proof sketch]
									Consider the reduction from 3SAT to an assignment setting from~\citep{AGM+15b,AGM+15c}. We show that checking whether the truthful preference profile is in PNE is coNP-complete. The problem is in coNP, since a Nash deviation is a polynomial time checkable No-certificate.
									The original reduction considers one manipulator (agent 1) while the other agents $N\setminus \{1\}$ are `non-manipulators'. 
							In the original reduction, the utility functions of agents in $N'=N\setminus \{1\}$ are not specified. We specify the utility function of agents in $N'$ as follows: the utility of an agent in $N'$ for his $j$-th most preferred house is ${(8n)}^{m-j+1}$, where $n=|N|$ and $m$ is the number of houses. These utility functions can be represented in space that is polynomial in $O(n+m)$. We rely on 2 main observations about the original reduction.
							First, in the truthful profile, whenever an agent finishes eating a house all houses have either been fully allocated or are only at most half eaten.
							Second, in the truthful profile every house except the prize house (the last house that is eaten) is eaten by at least 2 agents.
							We  now show that due to the utility function constructed, each agent from $N'$ is compelled to report truthfully. 		
 Assume for contradiction that this is not the case, and let us consider the earliest house (when running the PS rule) that some agent $i\in N'$ starts to eat although he prefers another available house $h$. Let $k$ denote the number of agents who eat a fraction of $h$ under the truthful profile. 
						By reporting truthfully, we show that agent $i$ can get $\frac{1/n-1/{2n}}{2}=1/{4n}$  more of $h$ than by delaying eating $h$. 
						Let us consider how much additional fraction of $h$ agent $i$ can consume by reporting truthfully. If he reports truthfully, he can start eating $h$ earlier and, in the worst case, he can only start $1/2n$ time units earlier. This means that $h$ is consumed earlier by a time of $1/2n$ if $i$ reports truthfully. Consider the time interval of length $1/2n$ between the time when $h$ is finished when $i$ is truthful about $h$ and the time $h$ is finished when $i$ delays eating $h$. In this last stretch of time interval $1/2n$, $i$ gets $\frac{1}{k}\cdot \frac{1}{2n}$ of $h$ extra when he does not report truthfully.
						Hence by reporting truthfully, $i$ gets at least $\frac{1/n-1/kn}{2}$ more  of $h$ which is at least $1/4n$ since $k\geq 2$. Due to the utilities constructed, even if $i$ gets all the less preferred houses, he cannot make up for the loss in utility for getting only $1/4n$ of $h$. 
						
Now that we have established that the agents in $N'$ report truthfully in a PNE, it follows that the truthful preference profile is in PNE iff the manipulator's truthful report is his best EU response. Assuming that the  agents in $N\setminus \{1\}$ report truthfully, checking whether the truthful preference is agent $1$'s best response  was shown to be NP-hard.
We have shown that the agents $N'$ report truthfully in a PNE. Hence checking whether the truthful profile is in PNE is coNP-hard.
	\end{proof}

Next, we show that computing a PNE with respect to the underlying utilities of the agents is NP-hard.

%

	\begin{theorem}\label{th:verifyPNE-NPhard}
	Given agent's utilities, computing a preference profile that is in PNE under the PS rule is NP-hard.		
		\end{theorem}
			\begin{proof}
				The same argument as above shows that the  agents in $N'$ play truthfully in a PNE. Hence, a preference profile is in PNE iff agent $1$ reports his EU best response.
		It has already been shown that computing this EU best response is NP-hard~\citep{AGM+15b} when the other agents are $N\setminus \{1\}$ and report truthfully. 
		Thus computing a PNE is NP-hard.
				\end{proof}

\subsection{Case of Two Agents}

In this section, we consider the case of two agents since many disputes involve two parties. Since an EU best response can be computed in linear time for the case of two agents~\citep{AGM+15b,AGM+15c}, it follows that it can be verified whether a profile is a PNE in polynomial time as well.

We can prove the following theorem for the ``threat profile'' whose construction is shown in Algorithm~\ref{algo:2agent-DL-Nash}.

		\begin{theorem}\label{th:threat}
			Under PS and for two agents, there exists a preference profile that is in DL-Nash equilibrium and results in the same assignment as the assignment based on the truthful preferences. Moreover, it can be computed in linear time.
		\end{theorem}
		\begin{proof}
			The proof is by induction over the length of the constructed preference lists.
			The main idea of the proof is that if both agents compete for the same house then they do not have an incentive to delay eating it. If the most preferred houses do not coincide, then both agents get them with probability one but will not get them completely if they delay eating them. The algorithm is described as Algorithm~\ref{algo:2agent-DL-Nash}.
		


		We now prove that $Q_1$ is a DL best response against $Q_2$ and $Q_2$ is a DL best response against $Q_1$. The proof is by induction over the length of the preference lists.
		For the first elements in the preference lists $Q_1$ and $Q_2$, if the elements coincide, then no agent has an incentive to put the element later in the list since the element is both agents' most preferred house. If the maximal elements do not coincide i.e. $h\neq h'$, then $1$ and $2$ get $h$ and $h'$ respectively with probability one. However they still need to express these houses as their most preferred houses because if they don't, they will not get the house with probability one. The reason is that $h$ is the next most preferred house after $h'$ for agent $2$ and $h'$ is the next most preferred house after $h$ for agent $1$. Agent $1$ has no incentive to change the position of $h'$ since $h'$ is taken by agent $2$ completely before agent $1$ can eat it. Similarly, agent $2$ has no incentive to change the position of $h$ since $h$ is taken by agent $1$ completely before agent $2$ can eat it.
		Now that the positions of $h$ and $h'$ have been completely fixed, we do not need to consider them and can use induction over $Q_1$ and $Q_2$ where $h$ and $h'$ are deleted.
			\end{proof}

		The desirable aspect of the threat profile is that since it results in the same assignment as the assignment based on the truthful preferences, the resulting assignment satisfies all the desirable properties of the PS outcome with respect to the original preferences. 
	Since a DL best response algorithm is also an EU best response algorithm for the case of two agents~\citep{AGM+15c}, we get the following corollary.

		\begin{corollary}
			Under PS and for 2 agents, there exists a preference profile that is in Nash equilibrium for any utilities consistent with the ordinal preferences. Moreover it can be computed in linear time.
		\end{corollary}


		\begin{algorithm}[tb]
			  \caption{Threat profile DL-Nash equilibrium for $2$ agents (which also is an EU-Nash equilibrium) which provides the same allocation as the truthful profile.}
			  \label{algo:2agent-DL-Nash}
			\renewcommand{\algorithmicrequire}{\wordbox[l]{\textbf{Input}:}{\textbf{Output}:}}
			 \renewcommand{\algorithmicensure}{\wordbox[l]{\textbf{Output}:}{\textbf{Output}:}}
			 \footnotesize
				\textbf{Input:}  $(\{1,2\},H,(\succ_1,\succ_2))$\\
				\textbf{Output:} The \emph{``threat profile''} $(Q_1,Q_2)$ where $Q_i$ is the preference list of agent $i$ for $i\in \{1,2\}$.
			\algsetup{linenodelimiter=\,}
			  \begin{algorithmic}[1]
				   \footnotesize
		\STATE Let $P_i$ be the preference list of agent $i\in \{1,2\}$
		\STATE
		Initialise $Q_1$ and $Q_2$ to empty lists.
		\WHILE{$P_1$ and $P_2$ are not empty}
		\STATE Let $h= \text{first}(P_1)$ and $h'=\text{first}(P_2)$
		\STATE Append $h$ to $Q_1$; Append $h'$ to $Q_2$
		\STATE Delete $h$ and $h'$ from $P_1$ and $P_2$
		\IF{$h\neq h'$}
		\STATE Append $h'$ to $Q_1$; Append $h$ to $Q_2$;
		\ENDIF
		\ENDWHILE	

		\RETURN $(Q_1,Q_2)$.

			 \end{algorithmic}
			\end{algorithm}

\begin{figure*}[ht!]
\centering
\includegraphics[width=0.9\textwidth]{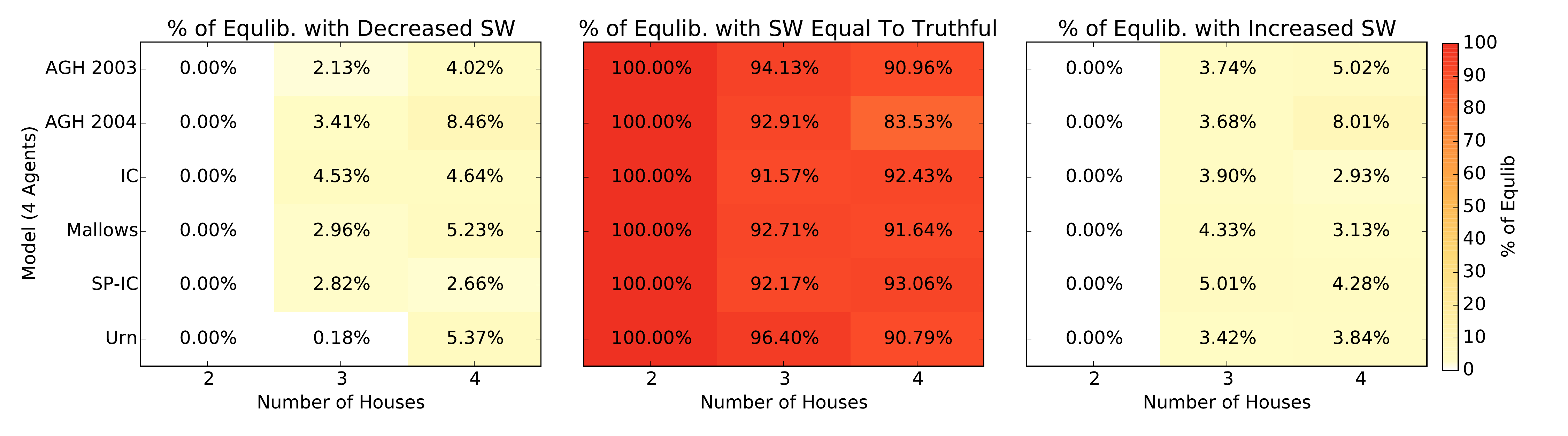}
\caption{Classification of equilibria 
in settings consisting of four agents with preferences drawn from the six models over 2 to 4 houses.
We can see that the vast majority of the equilibria that were 
found across all samples had the same social welfare as the
truthful profile. In general, we see there are roughly the same number of
equilibria that increase or decrease social welfare.}
\label{fig:percent-change}
\end{figure*}

\begin{figure*}[ht!]
\begin{minipage}[b]{0.65\textwidth}
	\centering
	\includegraphics[width=.9\textwidth]{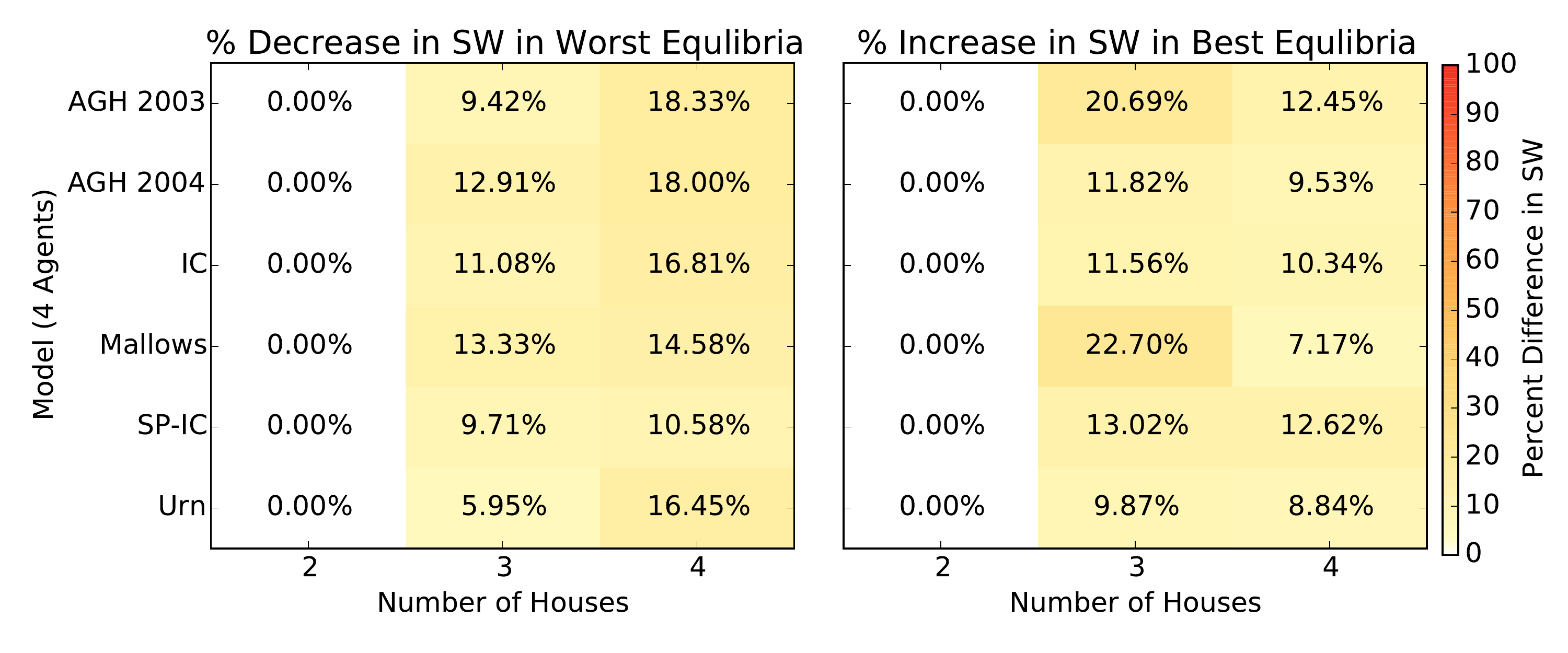}
	(A)
\end{minipage}
\hfill
\begin{minipage}[b]{0.33\linewidth}
	\centering
	\includegraphics[width=.9\textwidth]{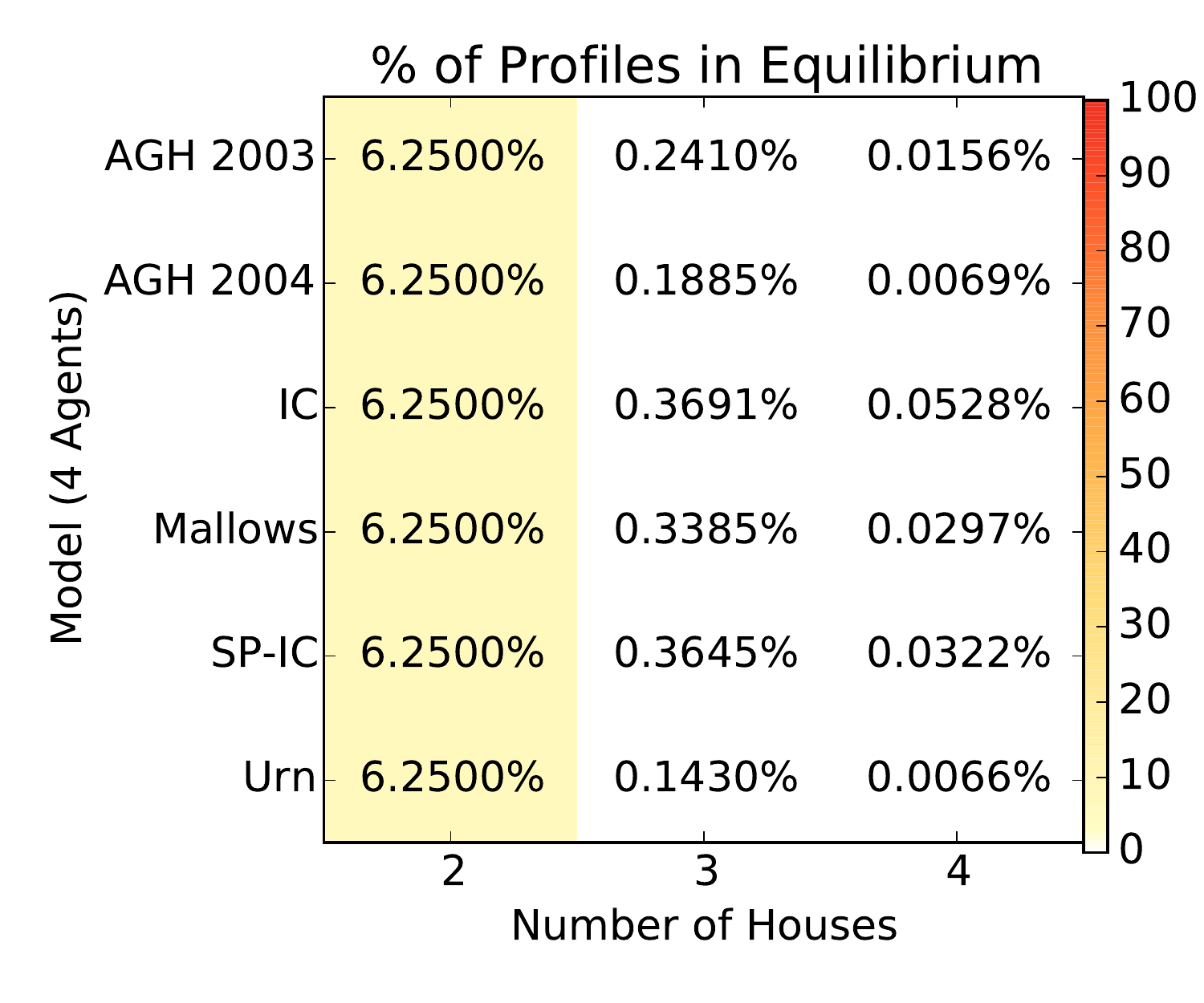}
	(B)
\end{minipage}	
\caption{(A) The maximum and minimum percentage increase 
or decrease in social welfare over all 300 samples for each of the six models with
2 to 4 houses and four agents.  We see 
that the gain of the best profile is, in general, slightly
more than the loss in the worst profile with respect to the truthful profile. (B) The average
number of the $4!^4=331,776$ profiles that are in 
equilibria per instance across all combinations of parameters.  The more uncorrelated
models (IC and SP-IC) admit the highest number of equilibria.}
\label{fig:increase-num}
\end{figure*}

\section{Experiments}


We conducted a series of experiments to understand the number and quality
of equilibria that are possible under the PS rule. 
For quality, we use the utilitarian social welfare (SW) function, i.e., the sum of the utilities of the agents.
We are
limited by the large search space needed to examine equilibria.  
For instance, for each set of cardinal preferences we generate, we
consider all misreports ($m!$) for all agents ($n$) leaving us with a search 
space of size $m!^n$ for each of the samples for each combination of parameters. 
Thus, we only report results for small numbers of agents and houses in this section.  
We generated 300 samples for each combination 
of preference model, number of agents, and number of items; reporting
the aggregate statistics for these experiments for 4 agents 
in Figures~\ref{fig:percent-change} and ~\ref{fig:increase-num}.
Each individual sample with 4 agents and 4 houses took about 15 minutes to complete using one core on
an Intel Xeon E5405 CPU running at 2.0 GHz with 4 GB of RAM running Debian 2.6.32.
The results for 2 agents and up to 5 houses as well as 3 agents and up to 4 houses are similar.

We used a variety of common statistical models to generate data (see, e.g.,~\cite{Matt11a,Mall57a,LuBo11a,Berg85a}):
the Impartial Culture (\textbf{IC}) model generates all 
preferences uniformly at random; the Single Peaked Impartial Culture
(\textbf{SP-IC}) generates all preference profiles that are single peaked
uniformly at random; Mallows Models (\textbf{Mallows}) is a correlated
preference model where the population is distributed around a 
reference ranking proportional to the Kendall-Tau distance; 
Polya-Eggenberger Urn Models (\textbf{Urn}) creates correlations
between the agents, once a preference order has been randomly
selected, it is subsequently selected with higher probability.
In our experiments we set the probability that the second order
is equivalent to the first to 0.5.
We also used real world data from
\textsc{PrefLib}
\cite{MaWa13a}: AGH Course Selection (ED-00009). This data 
consists of students bidding on courses to attend in 
the next semester.  We sampled students
from this data (with replacement) as the agents after 
we restricted the preference 
profiles to a random set of houses of a specified size. 

To compare the different allocations achieved under PS we need 
to give each agent not only a preference order but also a utility for each
house.  Formally we have, for all $i \in N$ and all $h_j \in H$, a value $u_i(h_j) \in \reals$.
To generate these utilities we use what we call the \emph{Random} model: 
we uniformly at random generate a real number between $0$ and 
$1$ for each house.  We sort this list in strictly decreasing order, if we cannot, 
we generate a new list (we discarded 0 lists in our experiments). 
We normalize these utilities such that each agent's utility sums to a constant value (here, the number of houses)
that is the same for all agents. In prior experiments we 
found the Random utility model to be the most manipulable and admit
the worst equilibria.
Therefore, we only focus on this utility model here 
(over Borda or Exponential utilities) as it represents, empirically, a worst case.
We separate equilibria into three categories: those where the SW is the same
as in the truthful profile, those where we have a decrease in SW, 
and those where we have an increase in SW.  
%
Given the social 
welfare of two different profiles, $SW_1$ and $SW_2$, we use percentage
change ($\frac{|SW_1 - SW_2|}{SW_1}\cdot 100$) to understand the magnitude of this difference.

For all models, for all combinations of 2 to 4 agents and 2 to 5 houses there are, generally, 
slightly more equilibria that increase social welfare compared to the truthful profile than those 
that decrease it, as illustrated in Figure~\ref{fig:percent-change}.  However, the vast majority
of equilibria have the same social welfare as the truthful profile, and the 
best equilibria are, in general, slightly better than the worst
equilibria, as illustrated in Figure~\ref{fig:increase-num}.  Hence, if any or all of the 
agents manipulate, there may be a loss of SW at equilibria,
but there is also the potential for large gains; and the most common outcome of all these 
agents being strategic is that, dynamically, we will wind up in an equilibria 
which provides the same SW as the truthful one.  
Our main observations are:
\begin{inparaenum}[(i)]
\item The vast majority of equilibria have social welfare
equal to the social welfare in the truthful profile.
\item In general, the number of PNE that have
increased social welfare (with respect to the truthful profile) is slightly more than the  
number of PNE that have decreased social welfare.
\item The maximum increase and decrease in SW in equilibria compared to the truthful profile was observed to be under 23\% and 18\% respectively .
\item There are very few profiles that are in equilibria, overall.  Profiles
with relatively high degrees of correlation between the preferences (Urn and AGH 2004) have fewer equilibrium profiles than the less correlated
models (IC and SP-IC).
\item These trends appear stable with small numbers of agents and houses.
We observed similar results for all combinations.
\end{inparaenum}

\section{Conclusions}

We conducted a detailed analysis of strategic aspects of the PS rule including the complexity of computing and verifying PNE.
The fact that PNE are computationally hard to compute in general may act as a disincentive or barrier to strategic behavior.
Our experimental results show PS is relatively robust, in terms of social welfare, even in the presence of strategic behaviour. 
Our study leads to a number of new research directions. 
It will be interesting to extend our algorithmic results to the extension of PS for indifferences~\cite{KaSe06a}.
Additionally, studying \emph{strong} Nash equilibria and a deeper analysis of Nash dynamics are other interesting directions.

\section{Acknowledgments}
NICTA is funded by the Australian Government through the Department of Communications and the Australian Research Council through the ICT Centre of Excellence Program. Serge Gaspers is also supported by the Australian Research Council grant DE120101761.

\bibliographystyle{named}
%

\end{document}